\documentclass[twoside]{article}

\usepackage[utf8]{inputenc}
\usepackage{fancyhdr}
\usepackage{natbib}

\usepackage[textwidth=6in,left=1.25in,right=1.25in,asymmetric]{geometry}
\usepackage{amsmath}
\usepackage{mathrsfs}

\usepackage{mathabx}

\usepackage{amsthm}

\usepackage{amsfonts}
\usepackage{amssymb}
\usepackage{graphicx}
\usepackage{placeins}
\usepackage{algorithm2e}
\usepackage{caption}
\usepackage{array}
\usepackage{multicol}
\usepackage{color}
\usepackage{mhequ}
\usepackage{bbm}
\usepackage[font={small,it}]{caption}
\usepackage{stackrel}
\usepackage{enumerate}
\usepackage{datetime}
\usepackage{breqn}
\usepackage{hyperref}
\usepackage{authblk}
\usepackage{multirow}

\theoremstyle{plain}

\newtheorem{theorem}{Theorem}[section]

\newtheorem{corollary}[theorem]{Corollary}

\theoremstyle{definition}
\newtheorem{definition}[theorem]{Definition}
\newtheorem{remark}{Remark}[section]

\newcommand\ci{\perp\!\!\!\perp}

\newcommand{\bz}{\mathbf{z}}

\newcommand{\bb}[1]{\mathbb{#1}}

\newcommand{\mc}[1]{\mathcal{#1}}
\newcommand{\bone}[1]{\mathbbm{1}\left\{ #1 \right\}}

\newcommand{\tx}{\widetilde{x}}
\newcommand{\tz}{\widetilde{z}}
\newcommand{\wh}[1]{\widehat{#1}}

\newcommand{\be}{\begin{equs}}
\newcommand{\ee}{\end{equs}}

\newcommand{\fair}{fair} 
\newcommand{\fairness}{fairness}
\newcommand{\tX}{\widetilde{X}}
\newcommand{\tZ}{\widetilde{Z}}
\newcommand{\tmu}{\widetilde{\mu}}
\newcommand{\tF}{\widetilde{F}}

\DeclareFontFamily{U}{matha}{\hyphenchar\font45}
\DeclareFontShape{U}{matha}{m}{n}{
  <-6> matha5 <6-7> matha6 <7-8> matha7
  <8-9> matha8 <9-10> matha9
  <10-12> matha10 <12-> matha12
  }{}

\DeclareSymbolFont{matha}{U}{matha}{m}{n}
\DeclareMathSymbol{\Lt}{3}{matha}{"CE}

\graphicspath{{./}}

\makeatletter
\providecommand*{\input@path}{}
\g@addto@macro\input@path{{./}}
\makeatother

\title{An algorithm for removing sensitive information: application to race-independent recidivism prediction}

\author[1]{James E. Johndrow \thanks{johndrow@stanford.edu}}
\author[2]{Kristian Lum \thanks{kl@hrdag.org}}
\affil[1]{Department of Statistics, Stanford University}
\affil[2]{Human Rights Data Analysis Group}

\date{\today}
\allowdisplaybreaks

\begin{document}

\maketitle

\begin{abstract}
Predictive modeling is increasingly being employed to assist human decision-makers. One purported advantage of replacing or augmenting human judgment with computer models in high stakes settings-- such as sentencing, hiring, policing, college admissions, and parole decisions-- is the perceived ``neutrality" of computers. It is argued that because computer models do not hold personal prejudice, the predictions they produce will be equally free from prejudice. There is growing recognition that employing algorithms does not remove the potential for bias, and can even amplify it if the training data were generated by a process that is itself biased.  In this paper, we provide a probabilistic notion of algorithmic bias. We propose a method to eliminate bias from predictive models by removing all information regarding protected variables from the data to which the models will ultimately be trained. Unlike previous work in this area, our framework is general enough to accommodate data on any measurement scale. Motivated by models currently in use in the criminal justice system that inform decisions on pre-trial release and parole, we apply our proposed method to a dataset on the criminal histories of individuals at the time of sentencing to produce ``race-neutral" predictions of re-arrest. In the process, we demonstrate that a common approach to creating ``race-neutral" models-- omitting race as a covariate-- still results in racially disparate predictions. We then demonstrate that the application of our proposed method to these data removes racial disparities from predictions with minimal impact on predictive accuracy.   
\end{abstract}
{\noindent \flushleft KEY WORDS: risk assessment, algorithmic fairness, criminal justice, neutral predictions, racial bias, recidivism, selection bias}

\section{Introduction} \label{sec:intro}

Statistical and machine learning models are increasingly used to inform high-stakes decisions, including hiring \citep{hoffman2015discretion}, credit scoring \citep{khandani2010consumer}, and throughout all stages of the criminal justice system. In the criminal justice context, predictive models of individuals' future behavior are used to inform judges regarding pre-trial release and bail setting, sentencing, and parole \citep{Phillips:2016aa, dieterich2016compas,brennan2009evaluating, cunningham2006actuarial, dvoskin2001risk, quinsey1998violent, berk2009forecasting}. For example, there is an increasing reliance on predictive models to inform judges about the likelihood that a defendant will re-offend if released. In this case, data about individual defendants is used to train a model with the objective of predicting future re-offense. The model's prediction for a given defendant is then shown to the judge or parole board presiding over that person's case for use in informing parole or sentencing decisions. Given the importance of decisions regarding an individual's personal liberty, it is imperative that any input to the decision-making process-- be it a model's prediction or otherwise-- be ``fair" with respect to legally or socially protected classes such as race, gender, sexual orientation, et cetera.

In the academic literature, there are several competing notions of algorithmic or model fairness, an overview of which can be found in \cite{Berk:2016aa}. In general, we have found that notions of fairness can be divided into three camps. One school of thought does not focus on a particular metric of fairness, but rather assumes a model will be fair if the protected variable(s) are omitted from the analysis. For example, a press release\footnote{http://mitsloan.mit.edu/newsroom/press-releases/mit-sloan-professor-uses-machine-learning-to-design-crime-prediction-models/} for a paper introducing a new recidivism risk assessment model \citep{zeng2015interpretable} acknowledges that recidivism prediction models can easily be mis-used to discriminatory ends if caution is not exercised. The authors purposefully excluded the race variable from a model intended to be used in sentencing because ``we do not want to punish people longer because of their race." The implication that omission of a protected variable will mitigate discrimination with respect to that variable is not unique to this case.\footnote{The focus of this paper was on making an interpretable model, not necessarily a fair one. In this case, the authors were following standard practice in this regard.} In fact, some proprietary software packages used in predictive policing models tout the fairness of their models on the basis of omission of a race variable \citep{Taylor:aa}. Using this procedure, if the permitted variables are correlated with the protected variables, even if the protected variables are omitted, their effects will remain in the estimated model via their correlation with the permitted variables. In the case of regression models, this is known as omitted variable bias \citep{clarke2005phantom}. 

The other two schools of thought acknowledge that algorithmic fairness is a non-trivial problem, but propose different remedies because they define fairness differently. One area of research defines fairness in terms of equivalence of some measure of predictive accuracy among all classes in a protected variable. For example, \cite{dieterich2016compas} argue that fairness is defined by similar accuracy and positive predictive value by class. \cite{zafar2016fairness} defines fairness in terms of equality of misclassification rates across class. A similar notion of fairness was proposed by \cite{hardt2016equality}, which argues that equivalence of false positive and false negative rates more accurately embodies everyday understanding of what it means to be fair. These two notions of fairness are indeed distinct. For example, \cite{chouldechova2016fair} shows the same positive predictive value by protected class and equal false negative and false positive rates cannot both be achieved when the outcome prevalence depends on protected characteristics. \cite{kleinberg2016inherent} shows theoretically that these notions of fairness are usually incompatible. A related literature focuses on methods and algorithms for achieving these notions of fairness by optimizing some utility or loss function subject to constraints that express the fairness criterion mathematically (e.g. \cite{dwork2012fairness}, which also applies to achieving the alternative notion of fairness we adopt here).

A third approach that has gained traction primarily in the computer science and machine learning literature defines fairness in terms of disparate impact on a protected class (see \cite{feldman2015certifying}, \cite{barocas2016big}). Under this definition, a model is typically considered fair if differences in the distribution of the model's predictions conditional on the protected variable do not exceed some pre-determined threshold, as measured by some appropriate notion of distance between probability distributions. In most cases, the allowable difference is zero, which is equivalent to the requirement that the predictive distribution is independent of the protected variable. This notion is sometimes called ``statistical parity'' or ``demographic parity.'' Our paper operates within this definition of fairness and builds upon the extant methodology for mitigating disparate impact.    

Methodology for achieving statistical parity has centered on removing information about the protected variables from the training data by transforming the training data. \citet{kamiran2009classifying} use a n\"{a}ive Bayes classifier to rank each observation in the training data by its probability of belonging to the ``desirable" category.\footnote{The authors actually optimize for a different notion of fairness that is nonetheless closely related to statistical parity.} Based on these rankings, the outcome variable in the training data is adjusted until there is no remaining association between the protected variable and the intended outcome variable. This procedure is limited to binary outcome data and the adjusted data is not re-usable in the sense that one could not then use the covariates to estimate relationships with other outcome variables and still be ensured of non-discriminatory outcomes. \citet{calders2010three} presents three algorithms for preventing a model from producing differential predictions by protected class by transforming the training data in accordance to an objective function that is minimized when the predictions from a model fit to the transformed data are independent of the protected variable. In this case also, the methods are restricted to binary protected classes and binary outcome variables. \citet{feldman2015certifying} propose a method for adjusting or ``repairing" the training data such that the user can tune the amount of permissible bias in models fit to the repaired training data. The authors suggest either removing information about the protected variable entirely or adjusting the training data such that that the differences in conditional predictive distributions cannot exceed the legal definition of disparate impact. One limitation of this approach noted by the authors is that only continuous-type covariates can be repaired. Further, it is not clear how this procedure could be used to protect a continuous variable without discretizing it, an approach that was taken in \cite{adler2016}.  A review and comparison of several more algorithms operating on binary protected and outcome variables can be found in \cite{romei2014multidisciplinary}. 

In this paper, we focus on recidivism prediction and seek to ``protect" a variable that encodes the individual's race\footnote{We use the terminology for racial categorization that is used in the dataset that is the focus of our application. The categories in this dataset are defined as African American, Caucasian, Hispanic, Asian, Native American, and Other}. The objective is to make predictions regarding an individual's future likelihood of re-offense that are fair with respect to that individual's race. Typically, post-release re-offense is measured by re-arrest, and there are many reasons to believe that re-arrest may be observed with bias with respect to race. For example, studies suggest that after controlling for criminal behavior, African Americans are more likely than Caucasians to become incarcerated \citep{bridges1988law}, and whether walking or driving, African Americans are disproportionately stopped and searched by police \citep{simoiu2016testing}. For drug crimes, African American drug users are arrested at a rate that is several times that of Caucasian drug users despite the fact that African American and Caucasian populations are estimated by public health researchers to use drugs at roughly the same rate \citep{langan1995racial}. Thus, fitting models to data for which certain racial groups are {\it observed} committing crime at a disproportionate rate unfairly biases the model's predictions against those racial groups. Given that the outcome variable is observed with bias with respect to the protected variable, we believe that the second set of approaches designed with the objective of achieving equivalent predictive accuracy by race are inappropriate for this particular application, as they ultimately rely upon comparing the model's predictions of re-offense to a fundamentally flawed and biased measure of re-offense: re-arrest.  The first class of approaches-- simply omitting race from the set of covariates used to fit the model-- is equally inadequate in this setting, as the covariates that are permitted to be used in the analysis are highly correlated with race. Thus, in this setting we advocate for the third school of thought-- that which seeks to obtain independence between race and model's predictions. By proposing this standard of fairness, we are arguing that in the absence of information to precisely quantify the differences in recidivism by race, the most reasonable approach is to treat all races as though they are the same with respect to recidivism. In this setting, it is implicit that among those who are classified as the same race, re-arrest is an unbiased measure of re-offense. While this may not be true, the methodology we propose allows one to protect as many variables as are necessary. That is, while we only protect the race variable in our application, we develop a statistical method that could be used to protect as many variables as  necessary to attain the legally or ethically relevant definition of fairness. 

The approach we suggest primarily builds on the work of \citet{feldman2015certifying}. In \citet{feldman2015certifying},
each variable, $X_j$, is transformed via a two-step transformation, in which the empirical quantiles of $X_j$ given the protected variable, $Z$, are mapped to the the quantiles of a distribution, $F_{A}$, which is defined to be the median of each of the conditional quantile distributions at each quantile. That is, the adjusted variable $\tilde{X}_j = F_A^{\leftarrow}(F_{\text{emp}}(X_j))$, where $F_{\text{emp}}$ is the empirical cumulative distribution function of $X_j$. The reliance on the empirical quantiles of each $X_j$ limits the scope of this approach to datasets for which each class in the protected variable has sufficient data such that the empirical distribution is a reasonable approximation to the true distribution. If some classes contain very few data points, this approach is inapplicable. This precludes protecting continuous variables and variables with relatively many classes, although in later work (\cite{adler2016}) the method is extended to better handle categorical variables.  A related fact is that this method is also limited in the sense that, in general, adjustments can only be made on a pairwise basis such that each variable in the newly created dataset is independent of the protected variable. This is insufficient to guarantee that all predictions generated from a model fit to the adjusted dataset will be independent of the protected variable, though in the examples on which the procedure is demonstrated, it works well empirically. Lastly, this approach does not allow for the  marginal adjustment of discrete $X_j$. This is obvious in the binary setting, where all ones would be mapped to the same value as would all zeroes, thus leaving the adjusted dataset fundamentally unchanged from the point of view of its correlation with the protected variable.

We instead approach the problem from a likelihood-based perspective. In this framework, extensions such as adjusting non-continuous variables as well as adjusting variables for which each class in the protected variable has little data become natural. This framework also allows us to make adjustments to mutual independence from the protected variable, rather than just pairwise independence. To do so, we define the problem in terms of a chain of conditional models, as is commonly used in multiple imputation (see \cite{white2011multiple} for an overview). Each variable is adjusted by matching its estimated conditional quantile (conditional on the protected variable and all other previously adjusted variables) to the marginal quantiles for that variable.  Whereas previous work has been limited to protecting only binary or categorical variables and adjusting a limited number of covariates, our approach allows for any number of mixed-scale variables to be adjusted. This greatly expands the range of datasets that can be adjusted and thus expands the universe of problems to which such adjustments may be applied. 

We apply our method to a dataset pertaining to the criminal justice system in Broward County, Florida. This dataset contains several covariates describing an inmate's demographic characteristics and criminal history. The outcome variable of interest is re-arrest within two years of release. We apply our proposed adjustment to the permitted covariates and use both logistic regression and random forest to predict re-arrest. We find that while models fit to the unadjusted data produce drastically different predictive distributions over the probability of re-offense by race -- thus empirically demonstrating the insufficiency of omitting race from the analysis when the goal is statistical parity -- equivalent models fit to the data adjusted using our procedure produce nearly identical predictive distributions by race. Further, the predictive accuracy of our method decreases only slightly due to the adjustment. We also find that random forest or logistic regression fit to only seven ``adjusted" variables -- mostly pertaining to an individual's criminal history -- has substantively equivalent predictive power to proprietary models used for recidivism prediction that use a battery of psychological questionnaires and evaluations in addition to information about the individual's criminal past. 

\section{Method for variable adjustment} \label{sec:method}
We describe a generally applicable method for creating an adjusted set of covariates that are independent of protected characteristics. We begin by presenting a probabilistic interpretation of the problem. Section \ref{sec:coupling} provides theoretical motivation for our proposed methodology. In section \ref{sec:univariate},  we present a method for creating maps that optimally transform {\it univariate} $X_j$ to $\tX_j$ such that all information about $Z$-- a potentially multivariate set of protected variables-- is removed from $\tX_j$. By applying this method to each $X_j$ independently, this section provides a way to achieve pairwise independence between each $X_j$ and $Z$. In section \ref{sec:chaining}, we extend the univariate methodology to allow for the adjustment of multivariate $X$ via conditional chaining to achieve mutual independence between $X$ and $Z$. 
 
\subsection{Setup}
Suppose we have a response $Y$ and predictors $(Z,X)$, where $Z$ represent protected characteristics. We take $X,Z$ to be $d_x$ and $d_z$ dimensional random vectors with arbitrary measurement scale. Consider a generic prediction rule or model for $Y$ given by 
\be
f : X \to \wh Y \label{eq:yhat}
\ee
Our goal is not to use any information about $Z$ in predicting $Y$; that is, we want a {\it fair prediction rule}.\footnote{We emphasize that the term ``fair'' is used here in a mathematical context as a shorthand for the independence condition in \eqref{eq:fair}, which is sometimes called statistical parity. Ultimately, it is up to policymakers and ethicists to determine whether this condition is appropriate in any particular context. However, we argue it is the most appropriate of the existing notions of algorithmic fairness to our motivating application for the reasons outlined in the introduction. } 
\begin{definition}[fair  prediction rule]
A prediction of the form \eqref{eq:yhat} is \emph{\fair} with respect to the protected characteristics $Z$ if and only if
\be
\widehat{Y} \ci Z. \label{eq:fair} 
\ee
\end{definition}

\noindent This definition of fairness places this work within the third of the enumerated schools of thought, and our mathematical characterization is consistent with the objective functions used in the works cited above that operate within this framework.
 
Although $f$ is not a function of $Z$ in \eqref{eq:yhat}, this is insufficient to guarantee $\wh Y \ci Z$ unless $X \ci Z$. In the overwhelming majority of applications, $X$ and $Z$ are dependent, and thus we must take additional measures to ensure $\wh Y$ is \fair. Without belaboring a point raised in the Introduction, a surprising number of researchers -- and we suspect the overwhelming majority of decision makers such as judges and parole boards -- mistakenly assume that any rule of the form \eqref{eq:yhat} is \fair. 

However, there is a simple condition that does guarantee \fairness. Since $Z$ does not appear in our model, we already have $\wh Y \ci Z \mid X$, so that $p(\wh y \mid z, x) = p(\wh y \mid x)$. Observe that
\be
p(\wh y \mid x) &= \int p(\wh y \mid x) p(x \mid z) dx, \\
p(\wh y) &= \int p(\wh y \mid x) p(x) dx
\ee
so that $p(x \mid z) = p(x)$ is sufficient for $p(\wh y \mid z) = p(\wh y)$, which is equivalent to $\wh Y \ci Z$. Thus, we seek to define a new random variable $\tX$ that is independent of $Z$, while still preserving as much ``information'' in $X$ as possible. The next section is concerned with defining $\tX$.

\subsection{Optimal coupling and transport maps} \label{sec:coupling}
Let $(\mc X,d)$ be a Polish space and $c: \mc X \times \mc X \to \bb R$ be a Borel ``cost'' function. Let $\mu,\tmu$ be probability measures induced by random variables $X,\tX$. The transportation distance with respect to $c$ is defined as
\be
g_c(\mu,\tmu) \equiv \inf_{\gamma \in \Gamma(\mu,\tmu)} \int c(x,\tx) d\gamma(x,\tx),
\ee
where $\gamma$ is a \emph{coupling} of $\mu,\tmu$ -- a joint distribution on $\mc X \times \mc X$ with marginals $\mu,\tmu$ -- and $\Gamma(\mu,\tmu)$ is the space of all couplings of $\mu,\tmu$. The transportation distance is the minimal total cost with respect to $c$ of transporting mass from $\mu$ to $\tmu$, and the coupling $\gamma^*$ achieving the minimal cost is the \emph{optimal coupling}, the solution to the Kantorivich transportation problem. In our context, $\gamma^*$ tells us how to find $\tX$ so as to minimize information lost, where information is quantified by $c$.

A natural choice in our setting is to set $c=d^q(x,\tx)$, with $d$ the Euclidean norm. If one later uses any method or algorithm based on linear functions of the covariates -- such as a generalized linear model -- making the Euclidean distance between the original and transformed covariates small will make the loss of predictive accuracy small. This logic can be extended to a broader class of methods
by applying our proposed procedure to nonlinear transformations of $x$. 

When $c(x,\tx) = d^q(x,\tx)$ for $q \ge 1$, the transportation distance is related to the Wasserstein-$q$ distance by $\mc W^q_q(\mu,\tmu) = g_c(\mu,\tmu)$, so the optimal coupling -- when it exists -- is also the coupling achieving the $\mc W_q$ distance. 

\subsection{Univariate transformations} \label{sec:univariate}
When $\mc X = \bb R$ and $d$ is the Euclidean norm, so that $\mu,\tmu$ are associated with distributon functions $F, \tF : \bb R \to [0,1]$,  
\be \label{eq:qtiletrans}
\mc W_q^q(F,\tF) = \int_0^1 |F^{\leftarrow}(p) - \tF^{\leftarrow}(p)|^q dp,
\ee
where $F^{\leftarrow}(p) \equiv \sup \{x \in \bb R : F(x) \le p\}$ is the left-continuous inverse of $F$ (\cite{dall1956sugli, mallows1972note, salvemini1943sul}, see also \cite{ekisheva2006transportation}). \eqref{eq:qtiletrans} does not require that $F$ is continuous. This result allows us to define the optimal coupling explicitly in the case where $F, \tF$ are absolutely continuous with respect to Lebesgue measure.

\begin{remark}[Optimal coupling on $\bb R$] \label{rem:optcouple}
 When $\mc X = \bb R$ with $d$ the Euclidean norm and $F, \tF$ have densities, the optimal coupling with respect to $c=d^q(x,\tx)$ for $q \ge 1$ is associated with the map $g(x) = \tF^{-1}(g^*(x))$ for $g^*(x) = F(x)$.
\end{remark}
\begin{proof}
 \be
 \bb E_F[c(x,g(x))] &= \int_{\bb R} \{x-\tF^{-1}(F(x))\}^q f(x) dx \\
 &= \int_{\bb R} \{F^{-1}(F(x))-\tF^{-1}(F(x))\}^q f(x) dx \\
 &= \int_{[0,1]} \{F^{-1}(p) - \tF^{-1}(p)\}^q dp.
 \ee
 So $g(x)$ achieves the transportation distance, and is therefore associated with the optimal coupling. 
\end{proof}
The proof of remark \ref{rem:optcouple} only required that $g^*(X)$ have a uniform distribution on the unit interval and $F^{\leftarrow}(g^*(X)) = X$ $F$-almost surely. This suggests how to achieve the Wasserstein distance using random maps when $F$ is atomic.

\begin{corollary}[Optimal coupling for atomic $F$ using stochastic maps ] \label{cor:optcoupleatomic}
 Suppose $\mc X = \bb R$ with $d$ the Euclidean norm and $F$ is atomic. Let $\dot{x} = \{\dot{x}_1,\dot{x}_2,\ldots\}$ be the support points of $F$ ordered such that $\dot{x}_j < \dot{x}_{j+1}$, with associated probabilities $\pi_j=\bb P[X = \dot{x}_j]$, and put $\nu_j = \sum_{j' \le j} \pi_j$.  Define a random map $g^*(X)$ by $g^*(X) \mid X=\dot{x}_j \sim \textnormal{Uniform}(\nu_{j-1},\nu_j)$, with $\nu_0 = 0$. Then the random map $g(X) = \tF^{\leftarrow}(g^*(X))$ achieves the optimal coupling.
\end{corollary}
\begin{proof}
 $g^*(X) \sim \text{Uniform}(0,1)$ marginally and $F^{\leftarrow}(g^*(X)) = X$ a.s. 
\end{proof}

In order to  achieve $g(X) \perp Z$ within the class of optimal transport maps above, we must have 
\be \label{eq:Tdefs}
g(X,Z) &= \tF^{\leftarrow}(g^*_{x \mid z}(X,Z)),
\ee
where $g^*_{x \mid z}$ is either the conditional distribution $F_{x \mid z}(X, Z)$ when $F$ is continuous, or is a random variable constructed as in Corollary \ref{cor:optcoupleatomic} with $\pi_j = \bb P[X = x_j \mid Z]$ when $F$ is atomic. This immediately implies an algorithm for transforming a univariate $X$ to $\tX $ such that $\tX \ci Z$ with minimal information loss, given in algorithm \ref{alg1}.

\RestyleAlgo{boxruled}
\LinesNumbered
\begin{algorithm}[ht]
  \caption{Univariate transformations of variables \label{alg1}}
  \KwData{$x = [x_1, ..., x_n]$, $\bz = [\bz_1, ..., \bz_n]$ where $x \mid \bz \sim F_{\bz}$, target distribution $\tilde{F}$ }
  \KwResult{$\tx = [\tx_1, ..., \tx_n]$, where $\tx \sim \tF$ and $\tx \ci \bz$.}
  \For{i = 1, ..., n}{
  \If{$X$ is atomic}{
  set $x_i^- = \max\{ \dot{x}_k : \dot{x}_k < x_i, k = 1, 2, ...\}$\\
  \If{$x_i^- =\emptyset$}{set $x_i^- = - \infty$} 
  set $l(x_i) = F_{\bz_i}(x_i^-)$; $r(x_i) = F_{\bz_i}(x_i)$ where  \\
  sample $u_i = \text{Uniform}(l(x_i), r(x_i))$ \\
  set $\tx_i = \tF^{\leftarrow}(u_i)$}
  \If{$X$ is continuous}{
   set $\tx_i = \tF^{\leftarrow}(F_{\bz_i}(x_i))$
  }
  }
\end{algorithm}

In the above, we have assumed that $F_{\bz}$ and $\tF$ are known. In practice, the conditional distribution of $X | Z \sim F_{\bz} $ is typically unknown and must be estimated from the data. In the example we present below, we have found traditional, parametric regression models to be successful at estimating $F_{\bz}$ if the analyst employs appropriate model selection and fit diagnostic techniques. A more automated approach would likely require more exotic non-parametric models to effectively model the conditional distributions without human input. In practice, how one chooses or estimates $\tF$ is less critical. As long as the size of the support of $\tF$ is at least as large as that of $F$, the specific choice $\tF$ does not affect the ranks of $\tx$. Thus any prediction rule that depends only on the ranks of the predictors will be invariant to choice of $\tF$. This includes regression-tree procedures, such as random forests. Moreover, it is typical in applied statistics and regression modeling to transform predictors prior to model fitting for computational reasons or to obtain better predictive accuracy, which would neutralize any choice we make for $\tF$. On balance, we suggest taking $\tF$ to be the marginal distribution $F_x$. This ensures that researchers using the transformed data still have access to the original marginal distribution of the data, which may be of significant value in its own right. 

In this section, we have demonstrated how to transform a univariate $X \rightarrow \tX$, where $\tX \ci Z$, with minimal information loss. This procedure could be applied independently to each covariate $X_j$ to achieve pairwise independence with $Z$, though it is not guaranteed that the resulting set of independently transformed covariates will achieve mutual independence with $Z$. 
Regardless, pairwise adjustments may be desirable to retain interpretability of the covariates. Under a pairwise adjustment, each covariate could be interpreted as a simple $Z$-adjusted version of itself. For example, if $X_j$ is the number of prior arrests and $Z$ is race, a pairwise adjustment would result in a race-adjusted measure of the number of prior arrests. The following section is concerned with extending the above results to achieve mutual independence with $Z$.

\subsection{Multivariate adjustments via chaining}\label{sec:chaining}

Now, let $X$ be a random vector. We propose to construct an analogous multivariate transformation of $X$ as

\be
\tx = g(X,Z) = (g_1(X_1,\tZ^{(1)}),g_2(X_2,\tZ^{(2)}),g_3(X_3,\tZ^{(3)}),\ldots,g_{d_x}(X_{d_x},\tZ^{(d_x)})), \label{eq:map} 
\ee
where $\tZ^{(j)} := \{Z,\tX_{1:j-1}\}$ for $j>1$ and $\tZ^{(1)} := \{Z\}$. The ordering $X_1,\ldots,X_{d_x}$ of the $X$ variables is arbitrary, though some orderings may be practically convenient for a given application.

Using basic rules of conditional probability, $p(\tx \mid z)$ can be decomposed as
\be
p(\tx \mid z) = \prod_j p(\tx_j \mid \tx_{1:(j-1)},z) = 
\prod_j p(g_j(x_j,\tx^{(j)}) \mid \tx^{(j)}).
\ee

From the above, it is clear that $g_j(X_j, \tZ^{(j)}) \perp \tZ^{(j)}$.  So for each element of the product, $p(g_j(X_j,\tZ^{(j)}) \mid \tZ^{(j)})$ can be replaced with $p(g_j(x_j, \tz^{(j)})) = p(\tx_j)$, and the joint distribution reduces to 

\be 
p(\tx \mid z) = \prod_j p(\tz_j)
\ee

\noindent and $\tX$ is mutually independent of $Z$.  Consequently, we refer to \eqref{eq:map} as a \emph{transformation to mutual independence}.

Although $\tX$ is independent of $Z$, it is not independent of $X$. The map in \eqref{eq:map} preserves information in $X$ by maintaining conditional ranks -- if $F_{x_j \mid \tz^{(j)}}(x_{j},\tz^{(j)}) > F_{x_j \mid \tz^{(j)}}(x_j', \tz^{(j)})$, then $\tx_j > \tx_j'$. That is, $\tx_j$ is a measure of how large $x_j$ is conditional on  $z$ and the other adjusted covariates. Informally speaking, $\tx_j$ is the part of $x_j$ which cannot be predicted by $\{z, \tx_{1:j-1}\}$. 

The key to this approach is reliably estimating the $g_j$, which in turn requires good estimators $\wh F^{\leftarrow}_{x_j}$ and $\wh F_{x_j \mid \tz^{(j)}}$. Estimation of $F_{x_j \mid \tz^{(j)}}$ is an exercise in regression modeling. When $Z$ is low dimensional, it may be appropriate to obtain the conditional through a nonparametric estimate of the joint $p(z,x_j)$. For the particular case of the recidivism data, likelihood-based parametric regression models were more successful. Ultimately, the better the estimator of the conditional distribution $\wh F_{x_j \mid \tz^{(j)}}$, the closer to fair any prediction rule $\wh y$ estimated on $\tx$, so it is critical to construct these estimators with care.  

\section{Simulation Example}
In order to illustrate our proposed method, we present a simple simulation example. We let $Z \sim \text{Bern}(0.5)$ be the protected variable. We sample two covariates, $X_1 \mid Z  \sim N(Z + 4, 1)$ and $X_2 \mid X_1, Z \sim \text{Pois}(\mu)$, where $\log(\mu) = -1 + \frac{1}{2} X_1Z + \frac{1}{10}X_1 + \frac{1}{6} Z$.  Finally, we define the dependent variable as $Y \sim N(2X_1 + X_2 + Z, 1)$. We simulate one realization from this model with sample size $n = 10,000$. Realizations from this model are denoted by lowercase variables, $z$, $y$, $x_1$, and $x_2$. A density plot of $y$ given $z$ is shown in Figure \ref{fig:sim-example-data}. The goal of our procedure is to make $f(\hat{y} \mid z=0)=f(\hat{y} \mid z=1)$ for a generic prediction rule $\hat{y}$. 

\begin{figure}[h]
\centering
\includegraphics[width=2in]{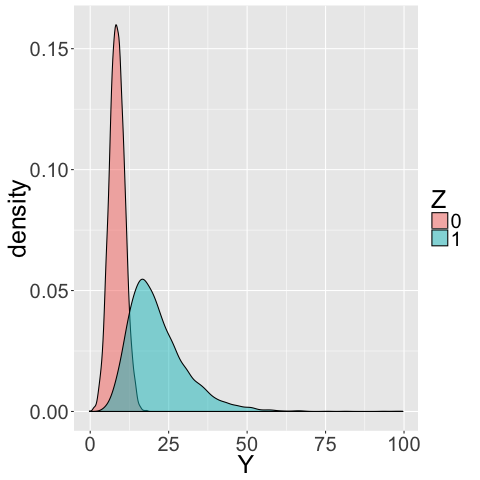}
\caption{\label{fig:sim-example-data} A comparison of the distributions of $y \mid z$ for each value of $z$ in the simulated data example.}
\end{figure}

We compare predictions of $y$ given $\tilde{x}$ by comparing three methods for producing $\tilde{x}$: (1) do no adjustment of $x$ (so $\tilde{x}=x$); (2) perform pairwise transformations to independence to produce $\tilde{x}$; and (3) perform transformations to mutual independence to produce $\tilde{x}$. In all cases, we omit $z$ from the set of covariates used to fit $y$. Fitted values are calculated as $\hat{y} = \tilde{x}\hat{\beta}$ where $\hat{\beta}$ is the least squares estimate from regression of $y$ on $\tilde{x}$. We compare the distribution of $\hat{y} \mid z$ under each adjustment procedure. 

For procedure 2, we must estimate both $\widehat{F}_{x_1 \mid z}$ and $\widehat{F}^{\leftarrow}_{x_1}$. To estimate the former, we fit a linear regression with Gaussian errors of $x_1$ on $z$, so that $\widehat{F}_{x_1 \mid z} = \Phi((x_1 - \hat{x}_1)/\hat{\sigma}^2)$, where $\hat{\sigma}^2$ is the estimated variance of $X_1 \mid Z$ and $\hat{x}_1$ is the fitted value of $x_1 \mid z$ under our estimated model. We set $\widehat{F}^{\leftarrow}_{x_1}$ to be the inverse empirical quantile function of $x_1$. Then, for each observation $i = 1, ..., n$, $\tx_{i1}= \widehat{F}^{\leftarrow}_{x_1}(\Phi(x_{i1} - \hat{x}_{i1})/\hat{\sigma}^2))$. 

Because $X_2$ is atomic, we use the procedure implied by Corollary \ref{cor:optcoupleatomic}. We first fit a Poisson regression of $x_2$ on $z$, resulting in fitted values $\hat{x}_2$. Under this fitted model,  $F_{x_2 \mid z}$ is the CDF of a Poisson distribution with rate parameter $\hat{x}_2$. We then sample $u_i \sim \text{Uniform}(l_i, r_i)$ for $l_i = \widehat{F}_{x_2 \mid z}(x_{i2}-1; \hat{x}_{i2})$ and $r_i = \widehat{F}_{x_{2} \mid z}(x_{i2}; \hat{x}_{i2})$. We again set  $\widehat{F}^{\leftarrow}_{x_2}$ to be the empirical quantile function of $x_2$, and $\tx_{i2}=  \widehat{F}^{\leftarrow}_{x_2}(u_i)$. This procedure is referred to as ``adjusted-pairwise". 

For procedure 3, we jointly adjust $x_1, x_2 \rightarrow \tx_1, \tx_2$. To do this, we make two transformations, the first by estimating the conditional distribution of $X_2$ given $Z$ and the second by estimating the conditional distribution $X_2$ given $\tX_1$ and $Z$.  For the first transformation, we use the same procedure as described above, so $\tx_1 = \widehat{F}^{\leftarrow}_{x_1}(\Phi(x_{i1} - \hat{x}_{i1})/\hat{\sigma}^2))$. For the second, we estimate $\widehat{F}_{x_2 \mid \tx_1, z}$ by fitting a Poisson regression of $x_2$ on $z$ and $\tx_1$, yielding fitted values $\hat{x}_2$. Then $\widehat{F}_{x_2 \mid \tx_1, z}$ is given by the CDF of a Poisson distribution with rate $\hat{x}_2$. Similar to the above, we make a stochastic transformation by sampling $u_i \sim \text{Uniform}(r_i, l_i)$ with $l_i = \widehat{F}_{x_2 \mid z, \tx_1^{(3)}}(x_{i2}-1; \hat{x}_{i2})$ and $r_i = \widehat{F}_{x_{2} \mid z, \tx_1^{(3)}}(x_{i2}; \hat{x}_{i2})$. We again use the empirical quantile function of $X_2$ as $\widehat{F}^{\leftarrow}_{X_2}$ to obtain $\tx_{2i} = \widehat{F}^{\leftarrow}_{X_2}(u_i)$ We refer to this procedure as ``adjusted".  

Figure \ref{fig:sim-results} shows the empirical distribution of the fitted $\hat{y}$ under each of the three adjustment procedures. From this it is clear that the unadjusted model does little to equalize the distributions of $\hat{y}$ conditional on $z$. The pairwise adjustment does reduce some of the discrepancy in the predictive distributions of $\hat{y} \mid z$, but the distribution of $\hat{y} \mid z=1$ has a much longer right tail than the distribution of $\hat{y} \mid z=0$. In contrast, the joint adjustment results in $\hat{f}(\hat{y} \mid z=1) \approx \hat{f}(\hat{y} \mid z=0)$, achieving the goal of the procedure. 

\begin{figure}[h]
\centering
\includegraphics[width=1.75in]{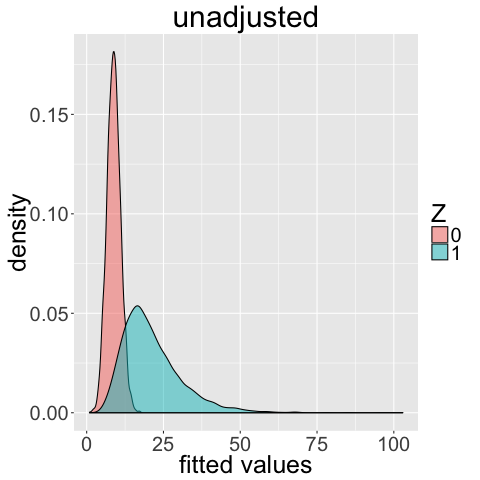}
\includegraphics[width=1.75in]{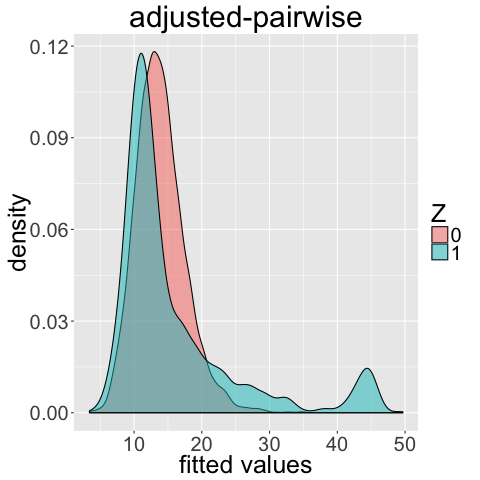}
\includegraphics[width=1.75in]{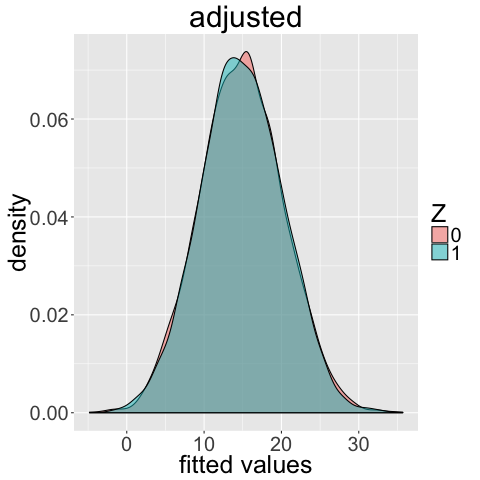}
\caption{\label{fig:sim-results}The distributions of $\hat{y} \mid z$ under each of the adjustment procedures.}
\end{figure}

\section{Application: Removing racial bias in recidivism risk assessment} \label{sec:application} 

ProPublica recently released a news article on the use of predictive analytics in recidivism risk assessment (\cite{angwin2016machine}). The focus of the the investigation was on whether risk assessment tools were disproportionately recommending non-release for African American defendants.  The reporters compiled an extensive dataset from the criminal justice system in Broward County, Florida, combining detailed individual-level criminal histories with predictions from a popular risk assessment tool, COMPAS. COMPAS (an acronym for Correctional Offender Management Profiling for Alternative Sanctions) is a proprietary software tool developed by Northpointe, Inc. that predicts a defendant's likelihood of failing to appear in court, re-offending, and violently re-offending. In order to produce the predictions, a proprietary algorithm is fit to several covariates, including a battery of psychological questions administered at the time of arrest. For each type of prediction, COMPAS produces a decile score (deciles of the predicted probability of re-arrest) and a categorical score consisting of three categories-- ``low", ``medium", and ``high" risk. In order to assess the accuracy of the recidivism predictions, the ProPublica researchers compared each person's COMPAS prediction to an indicator of whether they had been re-arrested within two years of release.

ProPublica's over-arching conclusion was that the COMPAS tool was ``racially biased" based on the observation that of those who were not re-arrested, 45 percent of African Americans were mis-classified by the model as future recidivists, where as only 24 percent of Caucasian defendants were similarly misclassified (\citet{angwin2016machine}). In a rebuttal, Northpointe asserted that the disparities in the proportion of false positives was entirely due to differing baseline rates of recidivism between African American and Caucasian defendants. They argued that bias should be assessed not in terms of the false positive rate, but rather, in terms of the group-wise positive predictive value or overall predictive accuracy. Using the same data used in ProPublica's analysis, Northpointe and others showed that the predictive accuracy of their model was equivalent for African American and Caucasian defendants \citep{dieterich2016compas, flores2016false}. ProPublica implicitly defined a prediction rule $\hat{Y}$ as ``biased" by race if

\be \label{eq:propublica}
\bb P(\hat{Y} = 1 \mid Y=0,  \text{African American})  \ne \bb P(\hat{Y} = 1 \mid Y=0, \text{Caucasian}). 
\ee
On the other hand, in one of their analyses Northpointe defined $\wh Y$ as ``biased" if 
\be \label{eq:northpointe}
\bb P(Y=1 \mid \hat{Y}=1, \text{African American}) \ne \bb P(Y = 1 \mid \hat{Y}=1, \text{Caucasian}), \ee
in other words, Northpointe adopted the equal conditional false positive rate notion of fairness. Clearly, when $P(Y=0 \mid \text{Caucasian}) \ne \bb P(Y= 0 \mid \text{African American})$, both notions of fairness cannot simultaneously be achieved, a fact that was demonstrated in \cite{chouldechova2016fair} and \cite{kleinberg2016inherent}. Because the African American population is measured as having a higher marginal recidivism rate, in order to achieve \eqref{eq:propublica}, the predictive accuracy of the model for African Americans would have to be higher than for Caucasians. Conversely, if we demand that \eqref{eq:northpointe} hold, the proportion of predictions that are false positives will necessarily be higher for African Americans. 

Ultimately, both definitions assume that re-offense is fairly measured by re-arrest. Given that the literature suggests that African Americans are more likely to be re-arrested for re-offending, neither definition seems particularly appropriate for this setting. One cannot actually know the extent to which the observation of re-offense is biased. Thus, as we argue above, a reasonable way to proceed is to assume that the distribution of risk is independent of race. In the absence of adequate evidence that the populations exhibit differing levels of re-offense post-release, one should revert to a ``null hypothesis" that the groups are the same in this regard. To this end, we implement the procedure described above to remove all information about race from the covariates we will use for prediction, thus guaranteeing similar distributions of estimated risk by race.

\subsection{Data}
For each defendent in the time period, ProPublica collected several measures of criminal history: the number of misdemeanor, felony, and other charges accrued as a juvenile (denoted respectively by juv\_misd\_count, juv\_fel\_count, juv\_other\_count); the number of adult prior offenses (prior\_count); the defendant's sex (sex); and age at the time of the crime (age). These are the covariates that make up $x$. The dataset also includes the race of the defendant (race), which is our protected variable, $z$. The response, $y$,  is an indicator of whether the defendant was re-arrested within two years of release. Using these data, the objective is to construct a new dataset $\tx$ that contains no information about $z$ so that any prediction rule of the form \eqref{eq:yhat} applied to the data set will satisfy  \eqref{eq:fair}.

\subsection{Dependence between race and other covariates in recidivism data}
We begin by assessing dependence between $z$ and $x$ in the data to determine whether transformations to independence are likely to have a meaningful effect. We test for pairwise dependence by discretizing continuous or count variables and summarizing data on pairs of variables in a two-way contingency table. We then compute the G statistic, 
$G(x_1,x_2) = 2 n \sum_{c_1=1}^{d_1} \sum_{c_2=1}^{d_2} \wh \pi_{c_1 c_2} \log[(\wh \pi_{c_1 c_2} )/(\wh \pi_{c_1 \cdot} \wh \pi_{\cdot c_2})]$, where $d_1$ and $d_2$ are the number of unique values of variables $x_1,x_2$, and $\wh \pi_{c_1 c_2} = n^{-1} \sum_{i} \bone{x_{i1} = c_1, x_{i2} = c_2}$,  $\wh \pi_{c_1 \cdot} = \sum_{c_2=1}^{d_2} \wh \pi_{c_1 c_2}$, and  $\wh \pi_{\cdot c_2} = \sum_{c_1=1}^{d_1} \wh \pi_{c_1 c_2}$ are the empirical cell probabilities of the contingency table. Clearly, $G$ is a scaled sample estimate of the mutual information between the joint distribution of the discretized variables and the product of their marginal distributions. The $G$ test is in fact a likelihood ratio test of the null hypothesis $H_0 : x_1 \perp x_2$ under the multinomial likelihood, and the test statistic has asymptotically a $\chi^2_{(d_1-1)(d_2-1)}$ distribution under the null hypothesis of independence.

We compute the $G$ statistic for all pairs of variables $(z,x_j)$ consisting of $z$ and one component of $x$. The $p$-values of the tests -- computed using the asymptotic distribution of the test statistic -- are shown in Table \ref{tab:gtestsorig}. There is strong evidence to reject the null hypothesis of independence for all of the pairs, even when adjusting for multiplicity using the method of \cite{benjamini1995controlling}. This indicates that a prediction rule $f : x \to \wh y$ is unlikely to be fair for race, and that to guarantee a fair prediction rule we need to estimate and apply a transformation to independence. Put another way, a model which simply excludes race is unlikely to result in fair predictions, as the effect of race will be encapsulated in the estimated effects of each of the variables included in the model. 

\begin{table}[ht]
\centering
\caption{p values for G tests of the null hypothesis of pairwise 
independence between race and the indicated variable, either unadjusted for multiplicity, or adjusted using the method 
of Benjamini and Hochberg (BH)} 
\label{tab:gtestsorig}
\begin{tabular}{rrr}
  \hline
 & unadjusted & BH \\ 
  \hline
sex & 8.84E-08 & 1.06E-07 \\ 
  juv\_fel\_count & 8.00E-21 & 1.20E-20 \\ 
  juv\_misd\_count & 1.91E-21 & 3.82E-21 \\ 
  juv\_other\_count & 2.72E-07 & 2.72E-07 \\ 
  priors\_count & 6.73E-58 & 4.04E-57 \\ 
  log(age) & 7.22E-49 & 2.17E-48 \\ 
   \hline
\end{tabular}
\end{table}

\subsection{Transformations to independence}
We now estimate maps of the form \eqref{eq:map} for each $x_j$ in the recidivism data. We first develop conditional density estimates $\wh F_{x_j \mid \tz^{(j)}}$ for each $x_j$. Of the six $x_j$, one (sex) is binary, one ($\log(\text{age})$ -- henceforth simply ``age'') is continuous, and the other four, which relate to prior criminal record, are counts. A pair plot, showing visualizations of the pairwise joint distributions of each of the covariates and race, is shown in Figure \ref{fig:marginals}.  The criminal record variables -- juv\_misd\_count, juv\_fel\_count, juv\_other\_count, and priors\_count -- are highly dispersed counts, and there is evidence of substantial dependence between most pairs of variables.   
\begin{figure}[h]
\centering
\includegraphics[width=\textwidth]{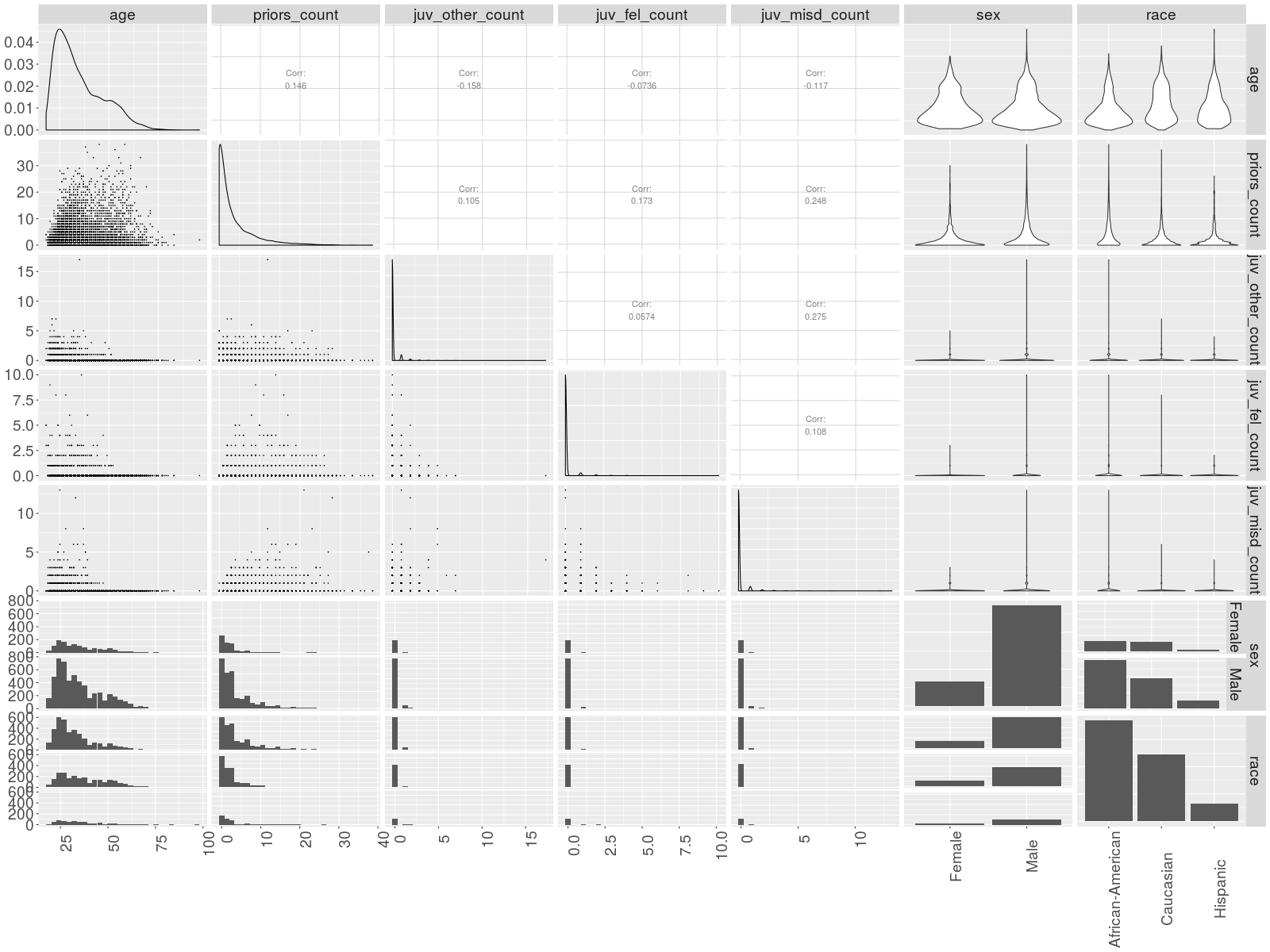}
\caption{Visualizations of marginal and pairwise conditional distributions of covariates. Row and column labels indicate which variables are compared. The diagonal shows marginal distributions; upper and lower triangles of the plot matrix show visualizations of the conditional distribution of row variables given column variables.} \label{fig:marginals}
\end{figure}

In constructing the sequence of conditional models, it makes sense to estimate $\wh F_{x \mid z}$ for the $x$ variables with more complicated marginal distributions first, which facilitates estimation of richer models. Based on Figure \ref{fig:marginals}, we order the variables as: $x_1 = \text{age}$, $x_2 = \text{prior\_count}$, $x_3=\text{juv\_other\_count}$, $x_4 = \text{juv\_fel\_count}$, $x_5 = \text{juv\_misd\_count}$, and $x_6 = \text{sex}$. The protected variable is $z = \text{race}$. We apply the procedure described in section \ref{sec:method} to estimate a transformation of the form \eqref{eq:map}. In every case, we estimate the marginal distribution of $x_j$, $\wh F_{x_j}$,  using the empirical CDF. In constructing the chain of conditional models, we always include discretized versions $\tilde{x}^*_1$ and $\tilde{x}^*_2$ of $\tilde{x}_1$ or $\tilde{x}_2$, respectively, whenever the continuous variable is included in the model. This captures nonlinearity in the conditional mean of the other variables. The cutpoints used for discretization are: $\{18, 19, 20, Q(0.1,\tx_1),Q(0.2,\tx_1),\ldots,Q(1,\tx_1)\}$ for $\tx^*_1$ and $\{Q(0.1,\tx_2),Q(0.2,\tx_2),\ldots,Q(1,\tx_2)\}$ for $\tx^*_2$, where $Q(p,x)$ is the $p$-empirical quantile of $x$. In every case, we estimate $\wh F^{\leftarrow}_{x_j}$ using the empirical quantile function of $x_j$, and we use our estimated $\wh F_{x_j \mid \tz_{j-1}}$ and $\wh F^{\leftarrow}_{x_j}$ and the data to obtain $\tx_j$ using Algorithm \ref{alg1}. 

The $\hat{F}_{x_j \mid \tz_{j-1}}$ are estimated as follows. 

\begin{enumerate}
\item Estimate $\wh F_{x_1 \mid z}$ using the empirical CDF of $x_1$ separately for each value of $z$. 

\item Set $\tz_1 = (z,\tx_1,\tx_1^*)$ and estimate $\wh F_{x_2 \mid \tz^{(1)}}$ by zero-inflated negative binomial regression of $x_2$ on  $\tz^{(1)}$.

\item Put $\tz_2 = (z,\tx_1,\tx_1^*,\tx_2,\tx_2^*)$ and estimate $\wh F_{x_3 \mid \tz^{(2)}}$ using a zero-inflated negative binomial regression of $x_3$ on $\tz^{(2)}$.

\item Estimate $\wh F_{x_4 \mid \tz^{(3)}}$ using a zero-inflated Poisson regression of $\tz^{(3)}$ on $\tx_4$. 

\item Estimate $\wh F_{x_5 \mid \tz^{(4)}}$ using zero-inflated Poisson regression of $\tz^{(4)}$ on $x_5$.

\item Estimate $\wh F_{x_6 \mid \tz^{(5)}}$ using logistic regression of $\tz^{(5)}$ on $x_6$. 
\end{enumerate}

We repeat the above $M$ times and save each of the transformed datasets, $\tx^{(m)} = \{ \tx_1, \tx_2, ..., \tx_6\}$ for $m = 1, ..., M$. Each resulting $\tx^{(m)}$ is stochastic because all of the $\wh F$ are discrete.  While any $(y,\tx)$ generated in this way is fair with respect to race, individual predictions depend on the sampled values $q(x) \sim \text{Uniform}(a(x), b(x))$ for all of the discrete variables, and interval estimates of parameters will understate uncertainty resulting from the stochastic nature of the maps $g_j$. Consequently, in generating predictive values for individual subjects or estimating uncertainty in model parameters, we use an average over all $M$ fair datasets $(y,\tx)$. This approach of creating multiple datasets is also used in the privacy settings \citep{reiter2005releasing} and multiple imputation \citep{rubin2004multiple, reiter2007multiple}, where a common default value is $M=10$ \citep{buuren2011mice}. In the fairness setting, we have the additional goal of limiting the effect of stochastic synthetic data $\tx$ on individual predictions, so we use a larger default value of $M=50$. 

If $\wh F_{x_j \mid \tz^{(j)}}$ were the exact conditional distribution $F_{x_j \mid \tz^{(j)}}$, then $\tx$ would satisfy $\tx \perp z$. Of course, $\wh F_{x_j \mid \tz^{(j)}}$ is an estimate, and thus it will differ from $\wh F_{x_j \mid \tz^{(j)}}$ in finite samples, and even asymptotically when $\wh F_{x_j \mid \tz^{(j)}}$ is misspecified. Therefore, we evaluate model fit for each conditional model separately, and recommend against applying a ``black box'' or automated approach to constructing the conditionals. We expect that in most applications, the number of predictors in $x$ will be relatively small, as is the case in our recidivism application, making it practicable to construct each conditional density estimate carefully. 

Because it is important that the entire conditional distribution is estimated well as opposed to just the conditional expectation, we assess model fit by plotting the fitted conditional CDFs $\wh F_{x_j \mid \tz^{(j)}}(x_j,\tz^{(j)})$ by race. If the fit is good, this should be close to the uniform distribution on the unit interval. Figure \ref{fig:qdens} gives results for the estimated conditional CDFs by race, all of which are approximately uniform. This indicates that we have sampled $\tx_j$ from the marginal distribution of $x_j$ for each race category. However, it is insufficient to guarantee that all information about race has been removed from the transformed dataset, as the model may be badly misspecified, e.g. not enough interaction variables were included in the model. We further analyze the success of the procedure's ability to achieve independence from the protected variable by computing Cramer's V statistics for every pair of variables in the adjusted data, discretized to have 10 unique values (or fewer, if the original variable is discrete with $<10$ unique values). Values in the original data are shown for comparison. In most cases, Cramer's V is reduced to near zero in the adjusted data, indicating that we have successfully removed information about race from the adjusted data, at least up to two-way interactions.

\begin{figure}[h]
 \centering
 \includegraphics[width=0.8\textwidth]{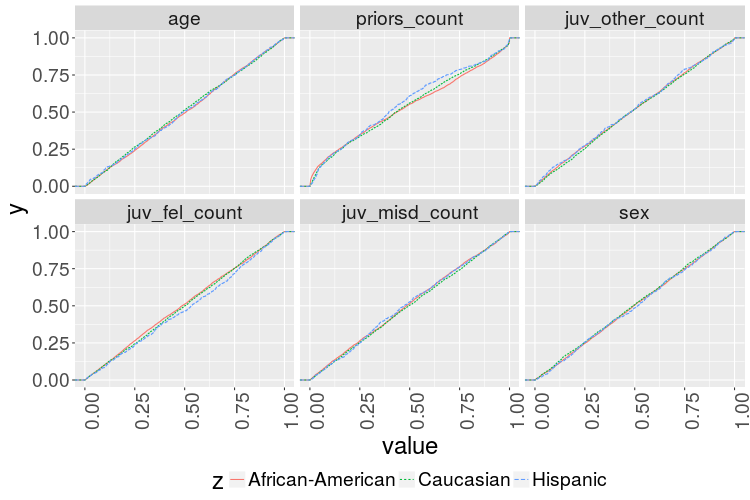}
 \caption{Plot of $F_{x_j \mid \tz^{(j)}}$ by race for each $x_j$ across 200 adjusted datasets} \label{fig:qdens}
\end{figure}

\begin{figure}[h]
\centering
\includegraphics[width=0.8\textwidth]{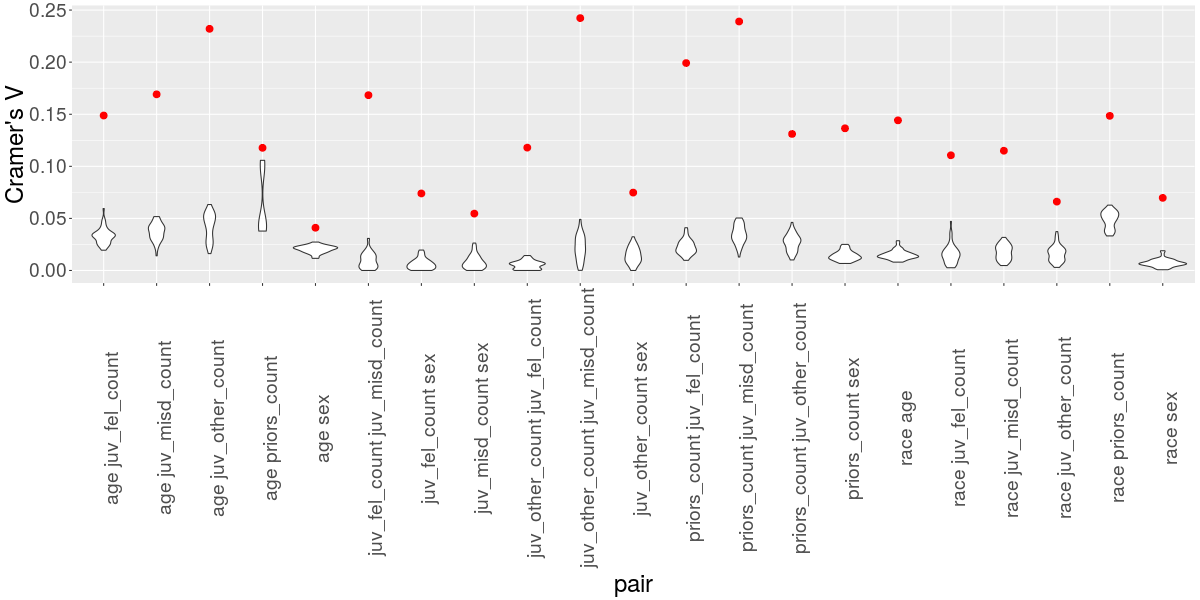}
\caption{Distribution of pairwise Cramer's V for every covariate pair across M transformed datasets for adjusted (white violin plot) and unadjusted datasets (red dot).} \label{fig:kspvals}
\end{figure}

\subsection{Predicting recidivism using transformed data} \label{sec:predictingRF}
Using each of the $M$ transformed datasets, we predict re-arrest within two years using random forest (RF).  We compare our results to the  ``unadjusted" model in which all covariates but race are used to explain $y$. We repeat this analysis using logistic regression in place of RF. The results of the logistic regression analysis are qualitatively very similar to those of RF and are deferred to the appendix.
Figure \ref{fig:cdf-rf} shows the empirical density and cdf of the re-arrest probability for RF trained on data adjusted using our procedure (``adjusted") and trained using the ``unadjusted" data.  It is clear from the left panels of Figure \ref{fig:cdf-rf} that when trained on unadjusted data, large differences by race exist in the predictive distribution, with the distribution for African Americans having substantially more mass at probabilities of re-arrest greater than about $0.5$. In other words, when trained on unadjusted data omitting race, the model predicts that a large fraction of the Africa-American population is at high risk of recidivating. Predictions made by training RF on data adjusted using our procedure eliminate almost all racial disparities, as evidenced by the nearly identical distributions by race in the two panels on the right.  

\begin{figure}[h]
\centering
\includegraphics[width=.95\textwidth]{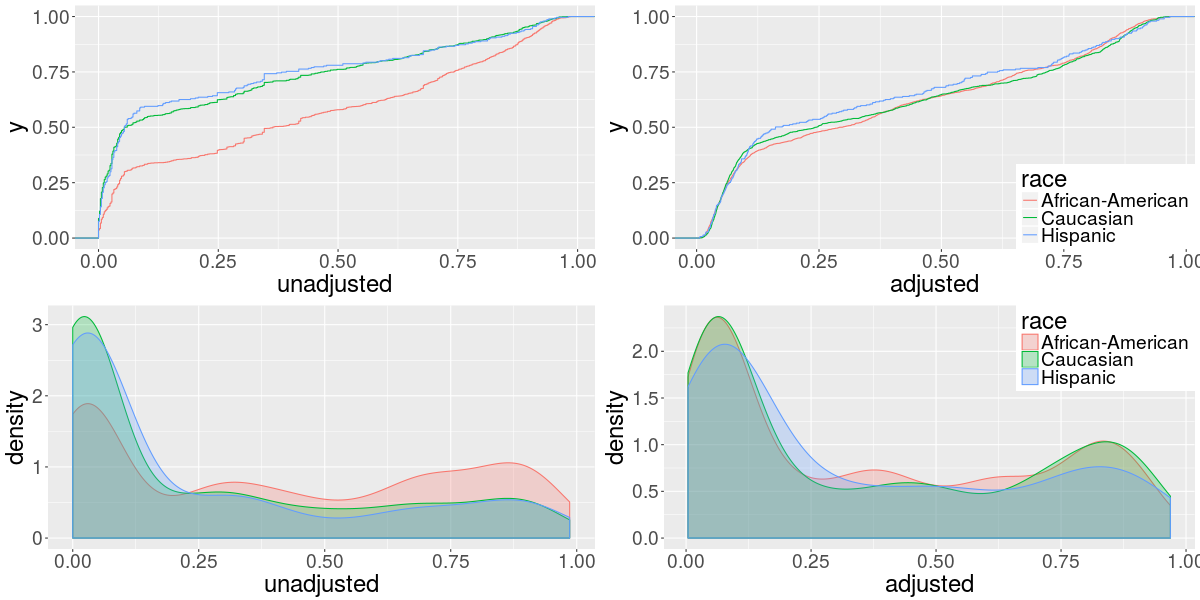}
\caption{density and cdf of predictions made using random forest by race using adjusted and unadjusted data} \label{fig:cdf-rf}
\end{figure}

Having established that predictions made using the transformed data are approximately {\it fair} under the definition we propose for this context, we now turn to fit assessment. In this case, assessment of how well our model predicts $Y$ using any notion of model performance is not especially well-motivated, as $Y$ is a biased measure of the phenomenon it is meant to measure. Nonetheless, we proceed to compare how well the predictions from RF fit to the unadjusted and adjusted datasets perform. In applying our procedure, some relevant information is lost. Thus, it is expected that the predictive accuracy of a model fit to the adjusted data will be lower than the model trained on unadjusted data. Figure \ref{fig:roc-rf} shows Receiver Operating Characteristic (ROC) curves for the predictions from the adjusted and unadjusted data. We find that these are not substantially different. For the unadjusted data, the area under the curve (AUC) was 0.72, and for the adjusted data, it was 0.71. We note that this AUC is on par with the AUC associated with Northpointe's predictions for this dataset (0.70) as reported in \cite{dieterich2016compas}.

\begin{figure}[h]
\centering
\includegraphics[width=0.5\textwidth]{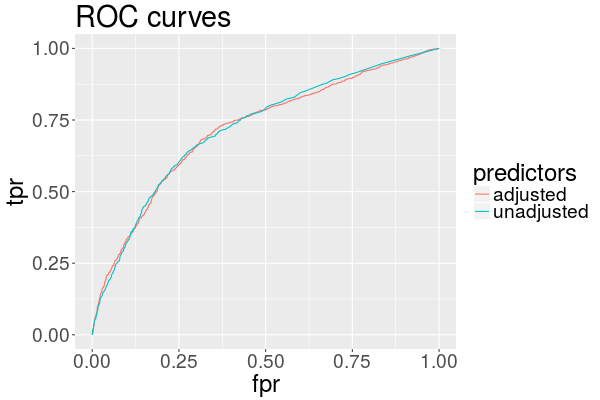}
\caption{roc for predictions made with random forest using adjusted and unadjusted data} \label{fig:roc-rf}
\end{figure}

Finally, we compare several measures of out-of-sample predictive performance across race under the adjusted and unadjusted models using $\hat{p}_i = 0.5$ as a threshold for classification. This is shown in Table \ref{tab:accuracy}, which reports accuracy (acc), positive predictive value (ppv), negative predictive value (npv), and false positive rate (fpr) for African Americans, Caucasians, and Hispanics.  The positive and negative predictive values exhibit disparities across race using both the adjusted and unadjusted data, but they are somewhat larger in the adjusted data. To make comparison easy, we show the mean absolute deviation (mad) using the median as the centroid, which increases after adjustment for both ppv and npv. In particular, the positive predictive value is reduced by adjustment for Caucasian and Hispanic prisoners and increased for African-Americans. Conversely, the negative predictive value is reduced by adjustment for African-Americans and increased by adjustment for Caucasians and Hispanics. 

On the other hand, adjustment appeared to decrease variation by race in false positive rates and only slightly increase variation in accuracy. The mad increases from 0 to 0.01 for accuracy, and decreases from 0.04 to 0.02 for fpr. The fpr is decreased for African-Americans and increased for Hispanics and Caucasians. Although compatibility (or lack thereof) between different notions of fairness has not been a focus of this paper, it is interesting that at least in this particular example, mitigating disparate impact actually led to an improvement in the overall similarity of accuracy and false positive rates by race, since variation in false positive rates fell considerably more than variation in accuracy increased after adjustment. Therefore, our procedure need not lead to a deterioration in fairness by all other metrics. This also suggests that optimizing a loss function that incorporates both similarity in false positive rates/accuracy and dependence of the predictive distribution on race may be sensible if deemed socially desirable by policymakers.

\begin{table}[ht]
\centering
\caption{Measures of predictive accuracy for Random Forest estimated on adjusted and 
             unadjusted datasets} 
\label{tab:accuracy}
\begin{tabular}{llllll}
  \hline
procedure & metric & African-American & Caucasian & Hispanic & mad \\ 
  \hline
Adjusted & ppv & 0.75 & 0.63 & 0.65 & 0.04 \\ 
  Unadjusted & ppv & 0.71 & 0.67 & 0.66 & 0.02 \\ 
  Adjusted & npv & 0.63 & 0.71 & 0.72 & 0.03 \\ 
  Unadjusted & npv & 0.64 & 0.68 & 0.68 & 0.01 \\ 
  Adjusted & acc & 0.67 & 0.68 & 0.7 & 0.01 \\ 
  Unadjusted & acc & 0.67 & 0.68 & 0.67 & 0 \\ 
  Adjusted & fpr & 0.18 & 0.22 & 0.19 & 0.01 \\ 
  Unadjusted & fpr & 0.24 & 0.13 & 0.12 & 0.04 \\ 
   \hline
\end{tabular}
\end{table}

\section{Discussion} \label{sec:conclude}

We have presented a statistical framework for adjusting a dataset such that models trained to the data will be mutually independent of protected variables. The framework we suggest has extended the existing literature by allowing an arbitrary number of variables of arbitrary type to be both protected and adjusted so long as a suitable conditional model can be found to adequately describe the full conditional distribution of the permitted variables given the protected variables. The extension to allow for the adjustment of discrete variables is itself an advancement, as previous proposals for adjusting of the {\it training covariates} were only designed for adjusting continuous variables. Our second main contribution is that our method allows the user to make adjustments such that the output dataset is mutually independent of the protected variables, as opposed to pairwise independent. We have tested this procedure on a dataset used for recidivism prediction and demonstrated that, by using a chain of relatively standard regression models, we are able to produce an adjusted dataset in which all pairs of variables in the dataset are approximately independent. Further, when fitting both random forest and logistic regression models to the data, we have achieved predictive distributions of recidivism that are approximately independent of race-- the ultimate goal of the procedure. Even after the adjustment, we observe that the quality of the predictions in terms of AUC are on par with methods that are currently in use but do not attempt to achieve fair predictions. We expect our procedure would also be of value in data privacy and anonymization.

It is often suggested that an equivalent way to accomplish the goal of removing disparate impact would be to simply take the top $x\%$ from each class and designate them as the most risky. However, adjusting the training data has other benefits. By doing the adjustment, a dataset could be released to multiple organizations to build prediction models, and regardless of the details of their model, we would be guaranteed that the predictions would be fair under the definition we support in this case. Additionally, if a protected variable is continuous (e.g. protecting parental income in a tool meant to predict success in college for college admissions), simply taking the top $x\%$ within each class is infeasible, as classes would have to be made by discretizing, and people who happened to fall on the lower or upper end of each bin would be unfairly disadvantaged. Lastly, if multiple variables are to be protected, even if each variable has sufficient data in each class, the combination of classes across all variables may not, necessitating an approach like that proposed here. 

There are several avenues for future work in this area. First, judges are typically the ultimate consumers of predictive risk assessment in criminal justice. In order to ease interpretability, predictions are often discretized into categories like ``recommended'' versus ``not recommended'' or ``high'', ``medium'', and ``low'' risk. Substantial research is necessary to better understand how the presentation of risk scores affects decision-making by judges. For example, decisions might differ if judges were shown the predicted probability  rather than a coarsened measure such as ``high risk." Moreover, if a coarsened risk score is presented, it is likely that the number of categories and the language used to describe each category would affect decisionmaking. In some cases, the highest risk group is in fact more likely than not to remain on good behavior post-release. It is unclear whether judges interpret ``high risk'' in those cases accurately, or implicitly assume that they are in fact more likely than not to re-offend. 

Another important issue that has heretofore received relatively little attention in the algorithmic fairness literature is the extent and severity of sampling bias in datasets used for predictive risk assessment, and whether this bias can be mitigated using additional data in some of these cases. As we argue in the case of recidivism, the likelihood that re-arrest is a biased sample of post-release criminal activity makes it difficult to assess the level of disparate impact, which is why we argue for the independence standard in this case. While obtaining additional data to mitigate bias seems implausible in this example, in other settings it may be more tractable, and there is need to develop appropriate methods to combine multiple samples in predictive risk scoring. 

Finally, it is imperative that we engage experts in other fields, as well as the communities most likely to be affected by the model's predictions, to aid in developing mathematical characterizations of fairness that aptly reflect the social or legal meaning of the term. This needs to be done separately in every context for which a predictive model is to be developed, as the best mathematical characterization of fairness will likely vary by context. While it is important that statisticians and others with related expertise take part in helping those outside our field understand proposed mathematical definitions of fairness -- for example, independence versus equality of false positive rate or positive predictive value -- ultimately this area of research should be undertaken in conjunction with ethicists and policy experts. At the moment, the conversation around these issues seems to be centered mainly in computer science, machine learning, and statistics, which certainly cannot result in an optimal outcome, since the issue of ``what is fairness?'' from a legal or ethical perspective clearly lies outside our area of expertise.

\section*{Acknowledgments}
This work was supported by the Human Rights Data Analysis Group. The authors thank David Dunson and Patrick Ball for helpful comments and conversations during the preparation of this manuscript. 

\newpage
\bibliographystyle{plainnat}
\bibliography{rnl}

\newpage

\begin{appendix}
\section{Appendix}
This section presents the results of applying logistic regression to predict $Y$. In an analysis that mirrors that presented in section \ref{sec:predictingRF}, we compare a logistic regression model applied to unadjusted data to one applied to data that has been adjusted under the procedure we propose. Figure \ref{fig:cdf-logreg} shows the cumulative distribution and density of the predictions by race. Like we found when using RF, omitting the race variable does little to reduce discrepencies in the distribution of predictions 
by race. However, a logistic regression model applied to the adjusted datasets result in very similar distributions of fitted values by race. Figure \ref{fig:roc-logreg} shows the ROC curves for each of the adjustment procedures. In this case also, there is little substantive difference between the methods in terms of this measure of predictive accuracy. 

\begin{figure}[h]
\centering
\includegraphics[width=0.95\textwidth]{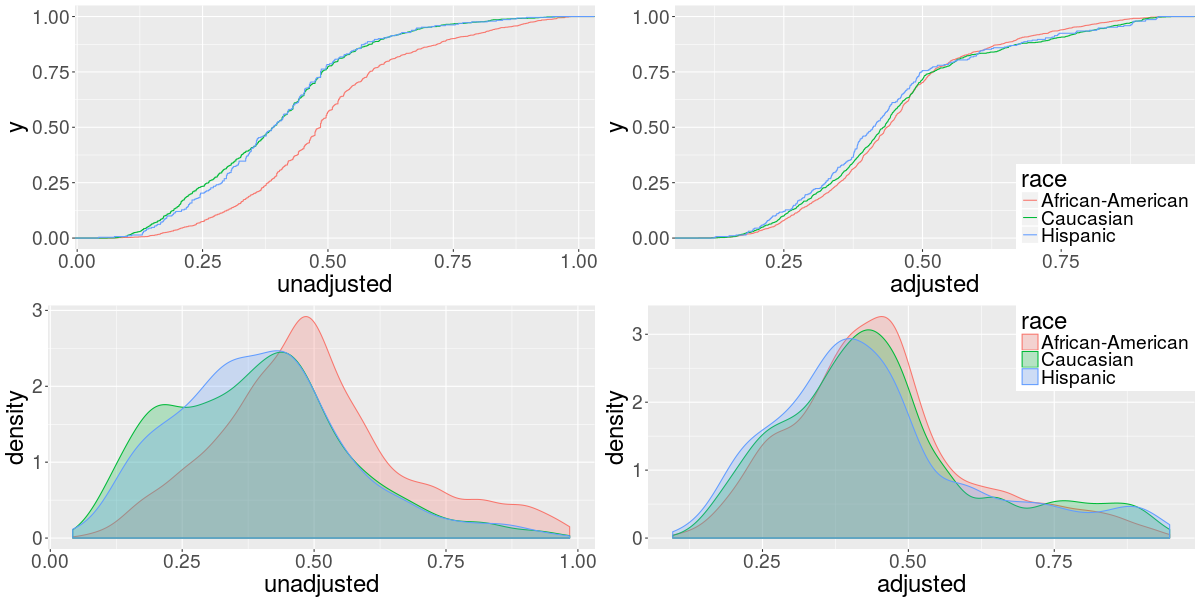}
\caption{The cumulative distribution (top)  and the density (bottom) of the predictions by race.} \label{fig:cdf-logreg}
\end{figure}

\begin{figure}[h]
\centering
\includegraphics[width=0.5\textwidth]{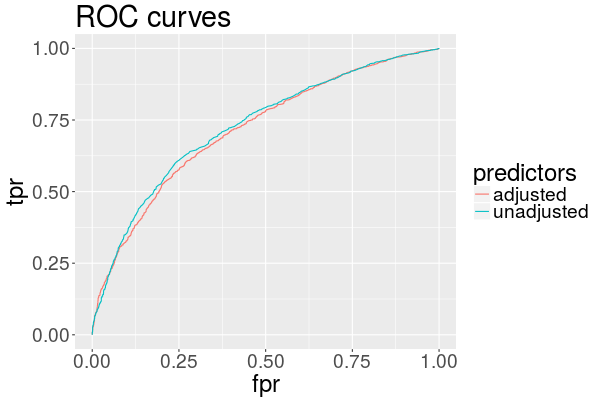}
\caption{ROC curve showing predictive performance of each of the adjustment procedures.} \label{fig:roc-logreg}
\end{figure}

\end{appendix}

\end{document}